\documentclass{l4dc2025}

\usepackage{times}
\usepackage{subcaption}
\usepackage{algorithmic}
\usepackage{url}
\usepackage{mathtools}

\usepackage{amsmath, amssymb, amscd, amsfonts}
\usepackage{graphicx}
\usepackage{hyperref}
\usepackage{enumerate}
\usepackage{xcolor}
\usepackage{float}
\usepackage{array}
\usepackage{lipsum}
\usepackage{multicol}

\usepackage{caption}
\usepackage{nomencl}
\makenomenclature

\title[Fast Stochastic MPC using Affine Disturbance Feedback Gains Learned Offline]{Fast Stochastic MPC using Affine Disturbance Feedback Gains  Learned Offline
}
\usepackage{times}

\coltauthor{\Name{Hotae Lee} \Email{hotae.lee@berkeley.edu}\and
 \Name{Francesco Borrelli} \Email{fborrelli@berkeley.edu}\\
 \addr 2169 Etcheverry Hall, Berkeley, CA, USA}

\begin{document}

\maketitle

\begin{abstract}%
We propose a novel Stochastic Model Predictive Control (MPC) for uncertain linear systems subject to probabilistic constraints. 
The proposed approach leverages offline learning to extract key features of affine disturbance feedback policies, significantly reducing the computational burden of online optimization. Specifically, we employ offline data-driven sampling to learn feature components of feedback gains and approximate the chance-constrained feasible set with a specified confidence level. By utilizing this learned information, the online MPC problem is simplified to optimization over nominal inputs and a reduced set of learned feedback gains, ensuring computational efficiency.

In a  numerical example  the proposed  MPC approach achieves comparable control performance in terms of Region of Attraction (ROA) and average closed-loop costs to classical MPC optimizing over disturbance feedback policies, while delivering a 10-fold improvement in computational speed. 

\end{abstract}

\begin{keywords}%
 Stochastic Model Predictive Control, Feature Extraction, Scenario Approach 
\end{keywords}

\section{Introduction}

 One key challenge in nominal Model Predictive Control (MPC) is the potential for constraint violations under model uncertainty. Robust MPC mitigates this by ensuring constraint satisfaction in worst-case scenarios; however, it often results in overly conservative policies with limited regions of attraction. Stochastic Model Predictive Control (SMPC) offers a more flexible alternative by employing chance constraints, which allow controlled constraint violations and yield less conservative policies~\cite{mesbah2016stochastic}.

Various SMPC methods address these challenges with different levels of conservatism. Techniques using Chebyshev’s inequality guarantee constraint satisfaction for any probability distribution but tend to be overly conservative \cite{cannon2012stochastic, korda2011strongly}. In contrast, methods tailored for specific distributions, such as Gaussian, offer improved performance but lack generality \cite{oldewurtel2008tractable, hewing2020recursively}. Sampling-based approaches handle arbitrary distributions, but their real-time application is hindered by the large sample sizes required, especially if affine disturbance feedback policies are used as decision variables \cite{zhang2013stochastic, hewing2019scenario, lee2023stochastic}. To overcome these limitations, offline scenario-based approaches leverage pre-sampled data to eliminate redundant constraints, reducing computational complexity \cite{lorenzen2017stochastic}. However, these approaches can still become infeasible for high-dimensional decision spaces or when dealing with large sample sets. Our work seeks to address these challenges building on the idea of sampling offline scenarios~\cite{lorenzen2017stochastic, mammarella2022chance} and proposing a technique which is  tractable when affine disturbance feedback policies are used as decision variables.

We propose a fast SMPC using an offline scenario approach that achieves performance comparable to traditional affine disturbance feedback SMPCs. The proposed method uses two key ideas:

\begin{itemize}
    \item 
Extracting features of affine feedback gains that have the highest effect on constraints  violation. This is done by simply applying Singular Value Decomposition (SVD) to the stacked constraint and dynamics matrix. When such features are used as optimization variables in the online SMPC this leads to  drastical reduction of online computation.

\item Deriving an approximate constraint set of in the space of features  which  maintains a specified confidence level for chance constraint satisfaction, simplifying the structure for faster MPC optimization.
\end{itemize}
When both key ideas are employed, the resulting SMPC problem will have fewer number of variables and constraints.

\section{Problem Formulation}

We consider a linear time-invariant (LTI) system with additive uncertainty:
\setlength{\abovedisplayskip}{4pt}
\setlength{\belowdisplayskip}{4pt}
\begin{align}\label{eq:sys}
    x_{t+1} = Ax_t + Bu_t+w_t ,~x_0=x_S, 
\end{align}
where the system matrices $A\in \mathbb{R}^{n\times n}$ and $B \in \mathbb{R}^{n \times m}$ are known, $x_t \in \mathbb{R}^n, u_t \in \mathbb{R}^m, w_t \in \mathbb{R}^m$ denote the state, control input, the uncertainty at time $t$ with a compact support $\mathbb{W}$, respectively.

The system under consideration is subject to chance constraints on  states and inputs described as follows:
\begin{align}
    \mathbb{P}(C x_t \leq d) \geq 1-\epsilon, ~~~
    \mathbb{P}(H_u u_t \leq h_u) \geq 1-\epsilon,
\end{align}
where $C \in \mathbb{R}^{n_{c} \times n} $ and $H_u \in \mathbb{R}^{n_{cu} \times m} $.
Our goal is to design a controller to have the system track a reference signal from a given initial state $x_S$ while satisfying the given chance constraints. 
We are interested in synthesizing a stochastic Model Predictive Control (MPC) by repeatedly solving the following optimal control problem:
\begin{subequations}\label{eq:mpc_formulation}
\begin{align}
    \min_{u_{0|t},u_{1|t}(\cdot),\cdots,u_{N|t}(\cdot)} & \sum_{k=0}^{N} c(\bar{x}_{k|t},\bar{u}_{k|t})+V(\bar{x}_{N+1|t}) \\ 
\mathrm{s.t.} ~~&x_{k+1|t} = Ax_{k|t}+Bu_{k|t}+w_{k|t}, ~~\bar{x}_{k+1|t} = A\bar{x}_{k|t}+B\bar{u}_{k|t}, ~\\
&\mathbb{P}(Cx_{k|t}\leq d)\geq 1-\epsilon, 
~~\mathbb{P}(H_u u_{k|t} \leq h_u)\geq 1-\epsilon, \label{eq:chance_constr}\\
&\bar{u}_{k|t} = u_{k|t}(\bar{x}_{k|t}),~ \bar{x}_{N+1|t} \in \mathcal{X}_F, ~x_{0|t} = \bar{x}_{0|t}= x_t, ~~\forall k \in \{0,1,...,N\}
\end{align}
\end{subequations}
where $\bar{x}_{k|t}$, $x_{k|t}$ denote the nominal state and the predicted state for prediction step $k\in \{0,1,\dots,N+1\}$, obtained from $x_t$ by applying the predicted input policies $\{u_{0|t}, \dots, u_{k-1|t}(\cdot)\}$. Also, $\bar{u}_{k|t}$ denotes the corresponding nominal input and $w_{k|t}$ denotes the uncertainty. $c(\bar{x}_{k|t},\bar{u}_{k|t}), V(\bar{x}_{N+1|t}), \mathcal{X}_{F}$ denote the nominal cost, the terminal cost and the terminal set, respectively. 
There are two main challenges to solve \eqref{eq:mpc_formulation}, namely:
\begin{enumerate}[(I)]
    \item Optimizing over policies $\{u_{0},\dots,u_{T-1}(\cdot)\}$ is an infinite dimensional problem, and computationally intractable, in general.
    \item The chance constraints need to be reformulated so that they can be solved with a numerical algorithm. 
\end{enumerate}
To address (I), we restrict our control policy to an affine disturbance feedback parametrization, i.e., 
\begin{align}\label{eq:feedback_policy}
    u_{k|t}  = \sum_{i=0}^{k-1} M_{k,i|t}w_{i|t}+v_{k|t},
\end{align}
where $M_{k,i|t}$ is the planned feedback gain for the prediction step $k$, corresponding to $w_{i|t}$, and $v_{k|t}$ is a nominal input.
To address (II) we use the scenario approach~\cite{calafiore2006scenario} to reformulate the chance constraints \eqref{eq:chance_constr} as discussed next.  

\subsection{Scenario Approach}
The scenario-based approach transforms chance constraints into a finite set of sampled constraints. While this method provides significant advantages over approaches tailored to specific probability distributions, it often requires a large number of samples, making it unsuitable for real-time MPC applications. The increased number of constraints directly impacts the computational burden of solving the optimization problem. When coupled with the affine disturbance feedback policy \eqref{eq:feedback_policy}, the optimization problem, which involves $\{\{M_{k,i|t}\}_{i=0}^{k-1}, v_{k|t}\}_{k=0}^{N}
$, becomes particularly computationally demanding due to the high-dimensional decision space and the complexity of the constraint structure.
To address this issue, we propose an offline sampling method which:
1) Extracts feature decision variables impacting state evolution in constraints to create a feature feedback policy that satisfies the chance constraints.
2) Computes an approximate chance-constrained set using the probabilistic scaling approach with offline scenarios proposed in \cite{mammarella2022chance}.

This study extends the concept that chance constraints can be probabilistically satisfied via the scenario approach based on offline samples \cite{lee2023stochastic,lorenzen2017stochastic}. Confidence in satisfaction increases with more scenarios; see \cite{lorenzen2017stochastic, calafiore2006scenario} for details on scenario requirements.

\section{Offline Sampling approach for Extracting feature disturbance feedback gains}\label{sec:feature_truncation}

Given the control policy in \eqref{eq:feedback_policy}, the decision variables in the MPC problem of \eqref{eq:mpc_formulation} are the feedback gain matrices $M_{k,i|t}$ and the nominal input $v_{k|t}$ for $k=0,\dots,N$ and $i=0,\dots,k-1$ for each $k$. 
In the SMPC formulation $\{\{M_{k,i|t}\}_{i=0}^{k-1}\}_{k=0}^{N}$ does not affect the objective function and is solely utilized for constraint satisfaction. By evaluating how $\{\{M_{k,i|t}\}_{i=0}^{k-1}\}_{k=0}^{N}$ influences the constraints across offline samples and the system dynamics, we construct a large constraint matrix composed of these samples and the corresponding state dynamics. We apply Singular Value Decomposition (SVD) to the matrix, as in classical feature extraction techniques. By truncating components with zero singular values, we extract the key feature components of the disturbance feedback policy, without altering its impact on the constraint matrix.
In the following sections, we provide a detailed discussion on the feature extraction process. 

\subsection{Constraint matrix from dynamics and offline samples}
For notational brevity, $\mathbf{y}_t$ denotes a vectorization of all gain matrices $\{M_{k,0|t}, \dots, M_{k,k-1|t}\}_{k=1}^{N}$ at time step $t$ and $\mathbf{v}_t$ denotes a vector stacking up all nominal inputs, i.e., $[v_{0|t}^\top, \dots, v_{N|t}^\top]^\top$. Then, the scenario approximation of \eqref{eq:mpc_formulation} with $N_s$ samples is rewritten as:
\begin{subequations}\label{eq:mpc_simple_ver}
\begin{align}
    \min_{\mathbf{y}_{t},\mathbf{v}_{t}} & ~~ \frac{1}{2} \begin{bmatrix}
        x_{0|t} \\
        \mathbf{v}_t
    \end{bmatrix} ^\top \mathbf{H} \begin{bmatrix}
        x_{0|t} \\
        \mathbf{v}_t
    \end{bmatrix}  \\ 
&\mathrm{s.t.} \nonumber \\
&\bar{A} \mathbf{y}_t + \bar{B} \mathbf{v}_t + \bar{C} x_{0|t} +\bar{d} \leq 0 , \label{eq:sampled_stacked_constraints}\\
&x_{0|t} = x_t,
\end{align}
\end{subequations}
where $\mathbf{H}$ denotes the quadratic cost weight matrix and $\bar{A}, \bar{B}, \bar{C}, \bar{d}$ denote matrices stacking dynamics and constraints from \eqref{eq:mpc_formulation}, described as next. 
$\bar{A}, \bar{d}$ contain the offline uncertainty samples. 
The matrices $\bar{B}$ and $\bar{C}$ are constructed by stacking $N_s$ instances of $B_{0:N}, C_{0:N}$, respectively. All stacked matrices are constructed as follows:
\begin{align}
    \bar{A} = \begin{bmatrix}
        A_{0:N}^0 \\ \vdots \\ A_{0:N}^{N_s}
    \end{bmatrix}, 
    \bar{B} = \begin{bmatrix}
        B_{0:N} \\ \vdots \\ B_{0:N}
    \end{bmatrix},
    \bar{C} = \begin{bmatrix}
        C_{0:N} \\ \vdots \\ C_{0:N}
    \end{bmatrix},
    \bar{d} = \begin{bmatrix}
        d^0_{0:N} \\ \vdots \\ d^{N_s}_{0:N}
    \end{bmatrix},
\end{align}
$A_{0:N}^j, B_{0:N}, C_{0:N}, d_{0:N}^j$ are constructed as follows:
\begin{align}
    A^j_{0:N} = \begin{bmatrix}
        w_{0}^{j \top} \cdots w_{N}^{j \top}
    \end{bmatrix} \otimes \begin{bmatrix} 
    0 & \dots & 0
    \\ CB  && \vdots
    \\ \vdots & \ddots & 0
    \\    CA^{N-1}B & \cdots & CB
    \\ H_u& \cdots & 0
    \\ \vdots & \ddots & \vdots
    \\ 0& \cdots & H_u
    \end{bmatrix},  \label{eq:A_kron}
    B_{0:N} = \begin{bmatrix}
    CB & \dots & 0
    \\ \vdots & \ddots& \vdots
    \\    CA^{N}B & \cdots & CB
    \\ H_u& \cdots & 0
    \\ \vdots & \ddots & \vdots
    \\ 0& \cdots & H_u
    \end{bmatrix},
\end{align}
\begin{align}
    C_{0:N} = \begin{bmatrix}
        CA \\ \vdots \\ CA^{N+1}
        \\ 0
        \\  \vdots 
        \\ 0
        \end{bmatrix},
    d^j_{0:N} =  \begin{bmatrix}
    C & \dots & 0
    \\ \vdots & \ddots & \vdots
    \\    CA^{N} & \cdots & C
    \\ & 0 &
    \\ & \vdots &
    \\ & 0 &
    \end{bmatrix} \begin{bmatrix}
        w_0^j \\ \vdots \\w_{N}^j
    \end{bmatrix} - \begin{bmatrix}
        d \\ \vdots \\ d 
        \\ -h_u
        \\ \vdots
        \\ -h_u
    \end{bmatrix},
\end{align}
where $j$ denotes the index of the $j$-th samples and $\otimes$ denotes a Kronecker product.
In this Quadratic Programming (QP) problem \eqref{eq:mpc_simple_ver}, $\mathbf{y}_t$ only appears in \eqref{eq:sampled_stacked_constraints} constraints, not in the objective function. We want to extract features of $\mathbf{y}_t$ by  discarding redundant elements that have no effect on \eqref{eq:sampled_stacked_constraints} and solve the problem for only the feature components without changing the optimization results. We apply SVD based feature extraction techniques in the following section. 

\begin{remark}
The terminal constraint can be expressed in terms of \( \mathbf{v}_t \) and easily integrated into the framework without any changes to the procedure if a robust invariant set is used, as it is independent of \( \mathbf{y}_t \). Detailing such design is beyond this paper's scope.
\end{remark}
\vspace{-10pt}


\subsection{SVD based Feature Extraction}
To efficiently apply SVD to $\bar{A}$ which is a large matrix, we first perform SVD on each of its sub-components. From \eqref{eq:A_kron}, we aim to apply SVD to each matrix prior to its formulation as a Kronecker product. Here, $\mathbf{W} \in \mathbb{R}^{N_s \times n(N+1)}$ denotes the vertically stacked form of $[w^{j \top}_0 \cdots w^{j \top}_{N}]$ for all $N_s$ samples and $ U_{W}, \Sigma_{W}, V_{W}$ denote the resulting matrices of SVD of $\mathbf{W}$. Next, we apply SVD to the matrix of \eqref{eq:A_kron} as: 
\begin{align}
    \mathbf{W} = \begin{bmatrix}
        w^{0 \top}_0 \cdots w^{0 \top}_{N}\\
        \vdots \\
        w^{N_s \top}_0 \cdots w^{N_s \top}_{N}
    \end{bmatrix} = U_{W} \Sigma_{W} V_{W}^\top,
        \begin{bmatrix}
    0 & \cdots & 0
    \\ CB && \vdots
    \\ \vdots & \ddots & 0
    \\    CA^{N-1}B & \cdots & CB
    \\ H_u& \cdots & 0
    \\ \vdots & \ddots & \vdots
    \\ 0& \cdots & H_u
    \end{bmatrix} = U_{B} \Sigma_{B} V_{B}^\top
\end{align}
Then we can present $\bar{A}$ with SVD submatrices as below:
\begin{align}
    \bar{A} = \begin{bmatrix}
        U_{B}&&\\
        &\ddots&\\
        &&U_{B}
    \end{bmatrix} (\mathbf{W} \otimes \Sigma_{B}) 
    \begin{bmatrix}
             V^\top_{B}&&\\
        &\ddots&\\
        &&V^\top_{B}
    \end{bmatrix}    
\end{align}
In this formulation, the first matrix is constructed by stacking \( N_s \) instances of $U_{B}$ diagonally, and the third matrix is formed by stacking \( nN \) instances of \( V^\top_{B} \) diagonally. 
\begin{align}
    \mathbf{W} \otimes \Sigma_{B} =  U_{W} \Sigma_{W} V_{W}^\top \otimes \Sigma_{B} = (U_{W} \otimes I_{p}) \Sigma_{W} \otimes \Sigma_{B} (V^\top_{W} \otimes I_{q}),
\end{align}
where $p = N_s, q= nN$ and the equality holds according to $(A\otimes B) (C \otimes D) = (AC) \otimes (BD)$ from the mixed-product property of Kronecker product. Then, $\bar{A}$ can be expressed as: 
\begin{align}\label{eq:svd_A}
    \bar{A} \!=\! (I_{N_s}\! \otimes \!U_{B}) (U_{W} \!\otimes \!I_{p}) \Sigma_{W}\!\! \otimes\!\! \Sigma_{B} (V^\top_{W} \!\otimes\! I_{q}) (I_{nN} \!\otimes \!V^\top_{B}) 
\!=\! (U_{W} \!\!\otimes\!\! U_{B}) (\Sigma_{W} \!\!\otimes \!\Sigma_{B}) (V^\top_{W} \!\!\otimes\! V^\top_{B})
\end{align}
This decomposition leverages the structure of the SVD to simplify SVD of the large matrix \( \bar{A} \).

\subsection{Truncated feedback policy}\label{sec:trun_y}
Consider the term  \( \bar{A} \mathbf{y}_t \) in \eqref{eq:sampled_stacked_constraints} using the SVD form of $\bar{A}$ in \eqref{eq:svd_A}.
Geometrically  \( V^\top_{W} \otimes V^\top_{B} \)  acts as a rotation matrix.
Since \( \mathbf{y}_t \) represents a vector of decision variables that can take any values, \( V^\top_{W} \otimes V^\top_{B} \mathbf{y}_t \) also spans the same dimensional space. Given that \( \Sigma_{W} \otimes \Sigma_{B} \) is a diagonal matrix, certain elements of \( V^\top_{W} \otimes V^\top_{B} \mathbf{y}_t \) are zeroed out after matrix multiplication. Consequently, we can truncate the elements of \( V^\top_{W} \otimes V^\top_{B} \mathbf{y}_t \) corresponding to the zero singular values of \( \Sigma_{W} \otimes \Sigma_{B} \) (i.e., \( \sigma_i = 0 \)). 
\begin{remark}
   In theory, only the components with exact zero singular values are truncated. In practice, components with small singular values can be removed. This will affect chance constraints guarantees and new probability satisfaction bounds can be derived as a function of the truncation parameters. Derivations are beyond the scope of this paper.
\end{remark}
The remaining valid components of the policy \( \mathbf{y}_t \) after truncation, are expressed as follows:
\begin{align}\label{eq:obtain_truncated_policy}
    \mathbf{y}^{\mathrm{trun}}_t = \mathrm{trun}((\Sigma_{W} \otimes \Sigma_{B}) (V^\top_W \otimes V_B^\top) \mathbf{y}_t)
\end{align} 
where \( \mathrm{trun} \) is a function that removes the corresponding zero elements, retaining only the non-zero elements. In other words, if the \( N_{y} \)-dimensional vector $(\Sigma_{W} \otimes \Sigma_{B}) (V^\top_W \otimes V_B^\top) \mathbf{y}_t$ contains \( N_{\mathrm{zero}} \) zero elements, this operation can be represented as a multiplication by an \( (N_{y} - N_{\mathrm{zero}}) \times N_{y} \) matrix with ones on the diagonal. Formally, the entire process can be expressed as $ \mathbf{y}^{\mathrm{trun}}_t = P^{} \mathbf{y}_t $ where $P$ denotes the product of all matrices involved.  Using this truncated policy and rotation/reflection matrices truncated accordingly, we can rewrite \eqref{eq:sampled_stacked_constraints} as:
\begin{align}\label{eq:truncated_sampled_constr}
    (U_{W} \otimes U_{B})^{\mathrm{trun}}\mathbf{y}^{\mathrm{trun}}_t + \bar{B} \mathbf{v}_t + \bar{C} x_{0|t} +\bar{d} \leq 0, 
\end{align}
where $(U_{W} \otimes U_{B})^{\mathrm{trun}}$ is a truncated version of $(U_{W} \otimes U_{B})$.
The truncated policy significantly reduces the number of decision variables without altering the sampled constraints. We refer to this reduced policy as the ``feature feedback policy".
Although the number of decision variables decreases, the sampled constraints (i.e., The number of rows of 
$\bar{A}$ ) remains unchanged and is large. This will be address in the next section.


\vspace{-0.0em}
\begin{remark}
The matrix $P$  retains a full-rank structure, as both the truncation matrix and the Kronecker product of diagonal and unitary matrices are full rank. Consequently, once \(\mathbf{y}_t^{\mathrm{trun}}\) is computed through the optimization process, the corresponding feasible \(\mathbf{y}_t\) can be reconstructed. $\mathbf{y}^{\mathrm{recon}}_t$ denotes the reconstructed $\mathbf{y}_t$. This ensures the applicability of the feature feedback policy. 

When vectorizing all the feedback gain matrices, certain elements of $\mathbf{y}_t$ must be zeros due to time causality. For instance, at prediction step $k$, the elements corresponding to $M_{k,k+1|t}$ should be zero when we construct the \eqref{eq:obtain_truncated_policy}. Without enforcing these zero constraints, some reconstructed $\mathbf{y}_t$ may include non-zero values in these positions, violating time causality and resulting in infeasible control policies. To address this issue,  we enforce \( \mathbf{y}_{t,i} = 0 \) for all indices \( i \in \mathcal{I} \), where \( \mathcal{I} \) represents the set of indices that must satisfy these zero conditions. Consequently, the relationship \( \mathbf{y}_t^{\mathrm{trun}} = P \mathbf{y}_t \) can be reformulated as: $\mathbf{y}_t^{\mathrm{trun}} = \bar{P} \mathbf{\bar{y}}_t$
where \( \mathbf{\bar{y}}_t \) comprises only the elements \( \mathbf{y}_{t,i} \) for \( i \notin \mathcal{I} \).
If \(\bar{P}\) is a full-rank matrix, then for any \(\mathbf{y}^{\mathrm{trun}}_t\), a corresponding \(\mathbf{\bar{y}}_t\) can be reconstructed such that the original policy \(\mathbf{y}_t\) satisfies \(\mathbf{y}^{\mathrm{trun}}_t = P \mathbf{y}_t\). In this case, no additional constraints are necessary in the optimization problem for \(\mathbf{y}^{\mathrm{trun}}_t\). However, if \(\bar{P}\) is not full rank, it becomes essential to impose further restrictions on \(\mathbf{y}^{\mathrm{trun}}_t\) by adding equality constraints and introducing new decision variables as follows:
\begin{align}
    \mathbf{y}^{\mathrm{trun}}_t = \bar{P}_{\mathrm{basis}} c,
\end{align}
where \(\bar{P}_{\mathrm{basis}}\) consists of the basis columns of \(\bar{P}\), and \(c\) is a vector of decision variables. Since the truncated policy \(\mathbf{y}^{\mathrm{trun}}_t\) already has a significantly reduced dimensionality, incorporating a few additional equality constraints and decision variables has minimal impact on computational efficiency, even in the worst-case scenario.

\end{remark}



\subsection{Probabilistic Properties}\label{subsec:prob_after_trun}
In \cite{lorenzen2017stochastic}, the main theorem proves that if the number of scenarios $N_s$ satisfies $N_s \geq \frac{5}{\epsilon}(\ln\frac{4}{\delta}+d \ln\frac{40}{\epsilon})$ where $d$ denotes the dimension of the set element, the feasible constrained set of scenario approximate problem belongs to the feasible set of the original problem allowing violations with at most probability $\epsilon$ and confidence level at least $1-\delta$. Consequently,
the solution $\mathbf{y}_t, \mathbf{v}_t$ of \eqref{eq:mpc_simple_ver} satisfies the original problem \eqref{eq:mpc_formulation} with confidence level $1-\delta$, provided that the number of samples $N_s$ meets the condition $N_s \geq \frac{5}{\epsilon}(\ln\frac{4}{\delta}+d \ln\frac{40}{\epsilon})$. If the control policy is truncated by removing only components corresponding to exact zero singular values ($\sigma_i = 0$), the truncated variables $\mathbf{y}^{\mathrm{trun}}_t$ and $\mathbf{v}_t$ would still satisfy all the existing sampled constraints. This ensures that the original constraints are satisfied with the same confidence level $1-\delta$ when $N_s \geq \frac{5}{\epsilon}(\ln\frac{4}{\delta}+d \ln\frac{40}{\epsilon})$ is used, even when the truncated policy is used.


\section{Probabilistic scaling based set approximation}\label{sec:scaling_set}
Here $z_t$ denotes $[\mathbf{y}_{t}^\top, \mathbf{v}_t^\top , x_{0|t}^\top]^\top$ at time step $t$ and $z^{\mathrm{trun}}_t$ denotes $[\mathbf{y}^{\mathrm{trun} \top}_t, \mathbf{v}_t^\top , x_{0|t}^\top]^\top$ at time step $t$. 
We define the $\epsilon$-chance constrained set of $z_t$ next.
\begin{definition}[$\epsilon$-Chance Constrained Set of $z_t$] 
    \begin{align}
        \mathcal{Z}_{\epsilon} = \{z_t ~|~ \eqref{eq:chance_constr}, ~\text{holds},  ~\forall k \in \{0,\dots,N\}\}
    \end{align}       
\end{definition}
Also we define the sampled set for each sampled disturbance sequence $w=[w_{0|t}, \dots, w_{N|t}]$ as:
\begin{definition}[Sampled Set of truncated $z_t^{\mathrm{trun}}$] 
    \begin{align}
        \mathcal{Z}^{\mathrm{trun}}(w) = \{z_t^{\mathrm{trun}} | Cx_{k|t}(w) \leq d, H_u u_{k|t}(w) \leq h_u,  ~\forall k \in \{1,\dots,N\} \},
    \end{align}       
\end{definition}
where $x_{k|t}(w), u_{k|t}(w)$ denote the states and the inputs evaluated over $w$ and the reconstructed $z^{\mathrm{recon}}_t$ from $z_t^{\mathrm{trun}}$. 
Using the definitions, $\mathcal{Z}^{\mathrm{trun}}_{\bar{N}}$ for any integer $\bar{N}$ denotes $\bigcap_{j=1}^{\bar{N}} \mathcal{Z}^{\mathrm{trun}}(w^j)$, which is an intersection of multiple sampled sets. 
A general scenario approximation approach constructs an inner approximate set of $\mathcal{Z}_{\epsilon}$ with the $1-\delta$ confidence level, using finite sampled constraints. Through the proposed truncation process in Sec. \ref{sec:feature_truncation}, we compute the sampled set $\mathcal{Z}^{\mathrm{trun}}_{N_s}$ in $z^{\mathrm{trun}}_t$ space, ensuring any reconstructed $z_t$ from $z^{\mathrm{trun}}_t$ in $\mathcal{Z}^{\mathrm{trun}}_{N_s}$ remains feasible within the approximate sampled set \eqref{eq:sampled_stacked_constraints}. $z^{\mathrm{recon}}_t$ denotes the reconstructed $z_t$. 
However, as elaborated in Sec. \ref{sec:trun_y}, the sampled constraints still result in a large number of constraints.
To address this, we aim to construct an approximate set that simplifies the representation of the sampled constraints, while maintaining the same $1-\delta$ confidence level of the scenario approximation. We achieve this by using the probabilistic scaling approach \cite{mammarella2022chance}. To explain how the probabilistic scaling approach keeps the $1-\delta$ confidence with the truncated feature policy, we briefly introduce the probabilistic scaling approach first.    




\subsection{Probabilistic Scaling Approach}\label{subsec:scaling_approach}



The approach finds an approximate set of $\{\theta \in \mathbb{R}^{q} ~|~ \mathbb{P} (F(w)\theta \leq g(w)) \geq 1-\epsilon \}$ denoted as $\Theta_{\epsilon}$ where $F(w), g(w)$ are a matrix and a vector which are functions of uncertainty $w$, respectively. The approximated set is expressed as the Minkowski sum of a center and a $p$-norm ball. 
The main idea of the method involves two primary steps:
1) Design the center and the shape of the candidate set, defined by a $p$-norm ball and an affine transformation $H_{\theta} \in \mathbb{R}^{q \times q}$
2) Find a scaling factor $\gamma$ to ensure the set satisfies chance constraints with \( 1-\delta \) confidence, expressed as $\Theta(\gamma) = \theta_{\mathrm{c}} \oplus \gamma H_{\theta} \mathbb{B}_{p}^{q}$ where $\mathbb{B}_{p} = \{s \in \mathbb{R}^q ~|~ \|s\|_p \leq 1\}$.

Similarly, we aim to find an approximate set for the $\epsilon$-CSS of $z^{\mathrm{trun}}_t$ while preserving \( 1-\delta \) confidence level. In the following section, we provide detailed explanations of the two-step procedure, along with the necessary adjustments to the proposed approach, and demonstrate that the confidence level is maintained.

\subsection{Adjusted Probabilistic Scaling Approach}
First we decide the center and the shape matrix for $z_t^{\mathrm{trun}}$ as design parameters, following the existing methodology in \cite{mammarella2022chance}. Since this set serves as an initial candidate, it is not required to satisfy the $1-\delta$ confidence level. Therefore, the center $z_c$ and the shape $H$ can be computed from the small number of samples denoted as $N_{\mathrm{ini}}$. The center can be computed as Chebyshev center of the sampled set $\mathcal{Z}^{\mathrm{trun}}_{N_{\mathrm{ini}}}$ and the shape is determined in a form of the zonotope represented with the infinity norm. We compute the $H$ matrix ensuing that the $z_c \oplus H \mathbb{B}_{\infty}^{n_{\mathrm{trun}}}$ is the largest set included within $\mathcal{Z}^\mathrm{trun}_{N_{\mathrm{ini}}}$ where $n_{\mathrm{trun}}$ is the dimension of $z^{\mathrm{trun}}_t$. The problem to obtain $H$ is written as follows:
\begin{align}\label{eq:find_design_parameters}
    \max_{z_c, H} & ~\mathrm{vol}_p(H) ~~
\mathrm{s.t.}~ z_c + H\mathbb{B}_\infty^{n_{\mathrm{trun}}} \subseteq \mathcal{Z}^{\mathrm{trun}}_{N_{\mathrm{ini}}}
\end{align}
Once the initial design parameters $z_c, H$ are computed, we compute the $\gamma(w)$ for $N_\gamma$ samples as:
\begin{align}\label{eq:find_gamma}
    \gamma(w) = \max_{ S(\gamma)\subseteq \mathcal{Z}^{\mathrm{trun}}(w)} \gamma,
\end{align}
where $w$ is a realized disturbance sequence $[w_{0|t}, \dots, w_{N|t}]$,  $S(\gamma) = z_c \oplus \gamma H \mathbb{B}_\infty^{n_{\mathrm{trun}}}$.
We compute the maximal scaling that $S(\gamma)$ is a subset of $Z^{\mathrm{trun}}(w)$ for each realization $w$ in the truncated $z$. For the specific $\epsilon$ violation probability and $1-\delta$ confidence level, we need at least $N_{\gamma}$ samples as:
\begin{align}
    N_{\gamma} \geq \frac{7.47}{\epsilon} \ln \frac{1}{\delta}
\end{align} 
Then, for all $N_{\gamma}$ samples, we sort out the $\gamma$ in descending order. We have to choose the $r$-th smallest value $\gamma^\star$ to guarantee the $1-\delta$ confidence level according to the probabilistic scaling approach, where $r= [\frac{N_{\gamma} \epsilon}{2}]$. Then, as mentioned in Sec. \ref{subsec:scaling_approach}, $S(\gamma^\star)$ is an approximate set of the $\epsilon$-CSS of $z^{\mathrm{trun}}_t$. 
To check the proof that the $S(\gamma^\star)$ is 
 a subset of $\epsilon$-CSS of $z^{\mathrm{trun}}$ with the specific confidence, when using the lower bound of the number of samples, see \cite{mammarella2022chance}. In this paper, we affirm that applying the probabilistic scaling approach to truncated variables still maintains the original chance constraint satisfaction with the same specific confidence level. 

\begin{proposition}[Probabilistic Scaling with Truncation]
Assume that the same sample set from the truncation process of Sec. \ref{sec:feature_truncation} is employed for computing the scaling factor for the truncated variables. Then, all \( z^{\mathrm{recon}}_t \) reconstructed from \( z^{\mathrm{trun}}_t \in S(\gamma^\star) \), is contained within \( \mathcal{Z}_\epsilon \) with the same \( 1-\delta \) confidence level. Furthermore, the minimum number of samples required for the adjusted probabilistic scaling remains unchanged.
\end{proposition}
\begin{proof}
Let \(\mathcal{W}\) denote the sample set. The constraints evaluated at \(z^{\mathrm{trun}}_t\) for all samples in \(\mathcal{W}\) are identical to those evaluated at the corresponding \(z^{\mathrm{recon}}_t\), where \(P_{\mathrm{re}}\) is the matrix that reconstructs \(z^{\mathrm{recon}}_t\) from \(z^{\mathrm{trun}}_t\) (i.e., \(z^{\mathrm{recon}}_t = P_{\mathrm{re}} z^{\mathrm{trun}}_t\)). Let \(x_{k|t}(w,z)\) and \(u_{k|t}(w,z)\) denote the state and input at prediction step \(k\), respectively, evaluated with disturbance \(w\) and variable \(z\).
For every \(w \in \mathcal{W}\), the state constraints \(Cx_{k|t}(w, z^{\mathrm{recon}}_t) \leq d\) and input constraints \(H_u u_{k|t}(w, z^{\mathrm{recon}}_t) \leq h_u\) are satisfied. Consequently, for all \(z^{\mathrm{trun}}_t\) within the set \(z_c \oplus \gamma H \mathbb{B}_\infty^{n_{\mathrm{trun}}}\), the reconstructed \(z^{\mathrm{recon}}_t = P_{\mathrm{re}} z^{\mathrm{trun}}_t\) lies within the feasible set \(\{ z \mid Cx_{k|t}(w,z) \leq d, H_u u_{k|t}(w,z) \leq h_u \}\) for all \(w \in \mathcal{W}\).
As a result, each \(z^{\mathrm{recon}}_t\) is contained within \(\mathcal{Z}(w)\) for at least \(N_\gamma - \left\lfloor \frac{N_\gamma \epsilon}{2} \right\rfloor\) samples \(w \in \mathcal{W}\) according to \eqref{eq:find_gamma}. Therefore, \(z^{\mathrm{recon}}_t\) belongs to the approximate set of (\(\epsilon\)-CSS) of \(z_t\) using the probabilistic scaling approach.
Since all inclusion relations hold in the \(z^{\mathrm{trun}}_t\) space, the required number of samples in the adjusted probabilistic scaling approach to guarantee the confidence of chance constraint satisfaction remains unchanged from the original method. Consequently, each reconstructed \(z^{\mathrm{recon}}_t\) derived from elements of \(S(\gamma)\) resides within the \(\epsilon\)-CSS of \(z_t\) with \(1 - \delta\) confidence.

\end{proof}

\vspace{-2.0em}

\section{Online MPC with principal feature policy and approximate sets}\label{sec:online_search}

\subsection{Construct the online MPC problem with offline components}
We construct the modified MPC problem with the offline computed components described in the previous sections. The decision variables of the online MPC problem are the truncated policy gains $\mathbf{y}^{\mathrm{trun}}_t$ and the nominal inputs $\mathbf{v}_t$ at time step $t$. The constraints are the approximate set  $S(\gamma^\star)$ obtained from the adjusted probabilistic scaling approach. In summary, the modified MPC problem is constructed as:
\begin{align}\label{eq:mpc_online}
    \min_{\mathbf{y}^{\mathrm{trun}}_{t},\mathbf{v}_{t}} & ~~ \frac{1}{2} \begin{bmatrix}
        x_{0|t} \\
        \mathbf{v}_t
    \end{bmatrix} ^\top \mathbf{H} \begin{bmatrix}
        x_{0|t} \\
        \mathbf{v}_t
    \end{bmatrix}  
~\mathrm{s.t.}  ~~[\mathbf{y}^{\mathrm{trun}}_t, \mathbf{v}_t, x_{0|t} ] \in S(\gamma^\star), ~~x_{0|t} = x_t,
\end{align}

\subsection{Implementation details}
Infinity norm can be converted into linear inequalities in \eqref{eq:infinity_norm_implement}.
Since taking the inverse of $H$ is computed offline, it does not degrade the efficiency of online implementation and numerical stability. 
 Since $x_{0|t}$ is included in $z_t$, we have to update the last elements of $z_t$ with $x_t$ every time step.
\begin{align}\label{eq:infinity_norm_implement}
    z \in z_c \oplus \gamma H \mathbb{B}_{\infty} 
    \iff \frac{z - z_c}{\gamma} = Hs,~ \| s\|_{\infty} \leq 1  
    \iff   -\gamma \leq H^{-1}(z - z_c) \leq \gamma  
\end{align}
Following the typical MPC framework, the first element of $\mathbf{v}^*$,  is selected as the control input for the current time step.
The entire process of the proposed algorithm is described in Algorithm. \ref{alg:1}.   

\setlength{\textfloatsep}{0pt}
 \begin{algorithm}[ht]
 \caption{Fast stochastic MPC with the feature feedback policy using offline sampling}
 \begin{algorithmic}[1]
 \renewcommand{\algorithmicrequire}{\textbf{Input:}}
 \renewcommand{\algorithmicensure}{\textbf{Output:}}
 \REQUIRE The uncertainty samples $\{[w^j_{0},\dots,w^j_{N}]\}_{j=0}^{N_{\mathrm{s}}}$
 \ENSURE  Control policy $u_{t}(\cdot)$

 \textit{Offline} :
 
 \STATE Obtain a SVD based feature feedback policy in \eqref{eq:obtain_truncated_policy} using offline sampling.
 \STATE Compute the norm-based approximate set of chance constrained set using adjusted probabilistic scaling approach in \eqref{eq:find_design_parameters}, \eqref{eq:find_gamma}.
 
 \textit{Online at time step $t$} :

\STATE Construct the MPC problem \eqref{eq:mpc_online} with the offline computed components
 \STATE Solve \eqref{eq:mpc_online} and apply $v_{0|t}^\star$ to the system \eqref{eq:sys} and Repeat the online process at time step $t+1$.
 \end{algorithmic}\label{alg:1}
 \end{algorithm}
\vspace{-2.3em}

\section{Simulation Results}
In this section, we demonstrate the practical application of the proposed method through numerical simulations. We consider a lane keeping scenario with the safety distance and the following distance not to disturb traffic flow. In this scenario, an ego vehicle keeps the lane with a preceding vehicle which is randomly changing a speed. The objective is to keep the reference speed as possible, while ensuring a safety distance (i.e., minimum distance) and a following distance (i.e., maximum distance) up to the preceding vehicle.
The preceding vehicle has uncertainties of speed and acceleration. 

We consider an integrator model of the ego vehicle as follows:
\begin{align}
    \begin{bmatrix}   s_{t+1} \\    v_{t+1}    \end{bmatrix} =
    \begin{bmatrix}    1 & \Delta t \\ 0 & 1    \end{bmatrix}
    \begin{bmatrix}     s_{t} \\ v_{t}     \end{bmatrix} + 
    \begin{bmatrix}         0 \\ \Delta t    \end{bmatrix} a_t,    
\end{align}
where $s_t, v_t, a_t$ denote the position, speed, acceleration of the ego vehicle at time step $t$, respectively. 
The dynamic model of the preceding vehicle is written as:
\begin{align}
    \begin{bmatrix}   s_{t+1}^{\mathrm{env}} \\    v^{\mathrm{env}}_{t+1}    \end{bmatrix} =
    \begin{bmatrix}    1 & \Delta t \\ 0 & 1    \end{bmatrix}
    \begin{bmatrix}     s_{t}^{\mathrm{env}} \\ v_{t}^{\mathrm{env}}     \end{bmatrix} + 
    \Delta t  \begin{bmatrix}         w_{1,t} \\  w_{2,t}  \end{bmatrix} ,    
\end{align}
where $s^{\mathrm{env}}_t,v^{\mathrm{env}}_t, a^{\mathrm{env}}_t,  w_{1,t}, w_{2,t}$ denote the position, speed, acceleration, uncertainty in speed and acceleration of the preceding vehicle at time step $t$, respectively.
The uncertainty samples are chosen from an uniform distribution, which is not known by the controller.
So the entire system dynamics is described in the form of \eqref{eq:sys} where 

\begin{align*}
    A = \begin{bmatrix}
        1 & \Delta t & 0 & 0 \\
        0 & 1 & 0 & 0 \\
        0 & 0 & 1 & \Delta t \\
        0 & 0 & 0 & 1
    \end{bmatrix},~
    B = \begin{bmatrix}
        0 \\ \Delta t \\ 0 \\ 0
    \end{bmatrix},~
    w_{t} = \begin{bmatrix}
        0 \\ 0 \\ w_{t}^{\mathrm{s}} \\ w_{t}^{\mathrm{v}}
    \end{bmatrix}
\end{align*}
and $x_t = [s_t,v_t,s_{t}^{\mathrm{env}}, v_{t}^{\mathrm{env}}]^\top, u_t = a_t$. 
The safety constraints, the following traffic flow constraints and the acceleration constraints are
\begin{subequations}\label{eq:sim_constr}
    \begin{align}
        s_t^{\mathrm{env}}-s_t\ge d_{\mathrm{safe}}-(v_t^{\mathrm{env}}-v_t)\Delta t \\
        s_t^{\mathrm{env}}-s_t\le d_{\mathrm{keep}} \\
        a_{\mathrm{min}} \leq a_t \leq  a_{\mathrm{max}}
    \end{align}
\end{subequations}
We want to have \eqref{eq:sim_constr} constraints violated in a probability lower than $\epsilon$, as below:
\begin{align}
    \mathbb{P}( \eqref{eq:sim_constr} ~\text{is satisfied} ) \geq 1-\epsilon
\end{align}
We choose the objective function as $J = \sum_{t=0}^T(v_t-v_{\mathrm{ref,t}})^2+a_t^2$, in order to minimize the tracking error and regularize the inputs. 
We find MPC solutions to the optimal control problem \eqref{eq:mpc_formulation} with simulation parameters in Table. \ref{table:sim_setup}.

To assess the effectiveness of our proposed approach, we compare it against two other policies of scenario approach: 
a) the open-loop policy, b) the affine feedback policy fully computed online (full feedback policy, in short). We evaluate these policies based on two key criteria: 
(I)  the performance quality of solutions and (II) the computational complexity.

\begin{table}[h!]
\centering
    \begin{tabular}{|c| c|| c |c|| c| c || c|c || c|c|} 
 \hline 
   \multicolumn{10}{|c|}{Simulation Parameters  (All units are SI units)} \\  
 \hline
 $s_0$ & 6 & $v_0$ & 6 & $v_{\mathrm{ref}}$& 5   &$s_0^{\mathrm{env}}$ & 20 & $v_0^{\mathrm{env}}$ & 4.5\\  [0.1ex] 
 \hline

  $w_0^{\mathrm{s},\mathrm{env}}$ & U(-0.6,0.6) &  $w_0^{\mathrm{v},\mathrm{env}}$& U(-0.8,0.8)&$\Delta T$ & 0.2 & Max acc & 10 & min acc& -10 \\ [0.1ex] 
 \hline
  $N_{\mathrm{s}}$ & $10^4$ &$N_{\mathrm{\gamma}}$  &  2000 & $N_{\mathrm{MPC}}$& 5 & $N_{\mathrm{trial}}$ &50 & $N_{\mathrm{task}}$& 20 \\
 \hline
  $d_{\mathrm{safe}}$ & 10 &$d_{\mathrm{follow}}$  & 16 &&&&&& \\ [0.1ex]
 \hline
\end{tabular}
\caption{Simulation parameter setup}
\label{table:sim_setup}
\end{table}
\vspace{-0.5em}
\subsection{Performance of the control policy}
The performance of the proposed controller is evaluated in two aspects: 1) the area of Region Of Attraction (ROA) 2) The sum of closed-loop costs.

First, we compare the area of approximate Region Of Attraction (ROA) between the approaches after we applied the proposed MPC and the other MPC policies a), b) mentioned earlier. We try to solve the MPC problem with multiple discretized states as an initial state of the ego vehicle. This analysis highlights the regions within the state space where each policy can obtain a feasible solution to control the ego vehicle. The results are illustrated in Fig. \ref{fig:roa_all_4}. The blue circle denotes the feasible initial states, while the red cross denotes the infeasible initial states. The open-loop policy has an empty region of the feasible states, implying the scenario is too tight for finding a conservative policy satisfying the constraints. On the other hand, the proposed policy and the full feedback policy have some feasible regions. Remarkably, the proposed method boasts an ROA that closely rivals that of the full feedback gain policy. 
To evaluate the robustness of the proposed approach, we conduct a similar simulations using a less restrictive initial $s_0^{\mathrm{env}}$ which is $21$m. The results are presented in Fig. \ref{fig:roa_all_7}. While the open-loop policy exhibits a non-empty feasible region, both the proposed policy and the full feedback policy demonstrate larger feasible regions, with each being of comparable area.

\begin{figure}[h]

\centering
\subfigure[Open-loop policy]{\includegraphics[width=0.32\textwidth]{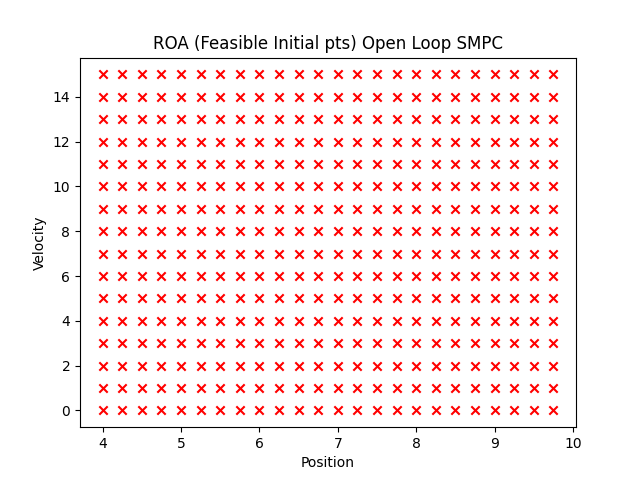}}
\subfigure[Full feedback policy]{\includegraphics[width=0.32\textwidth]{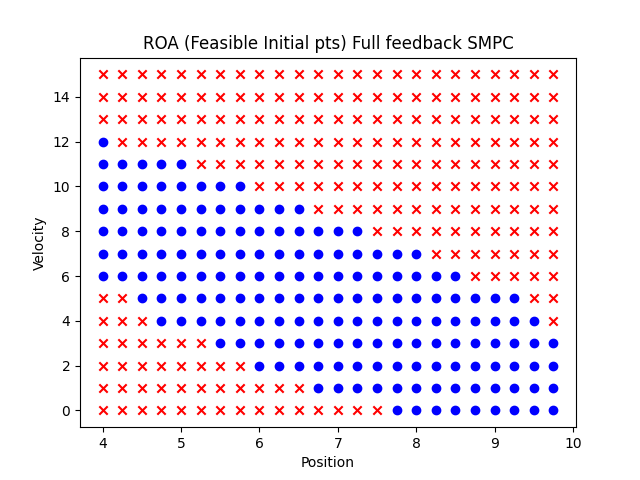}}
\subfigure[Proposed feedback policy]{\includegraphics[width=0.32\textwidth]{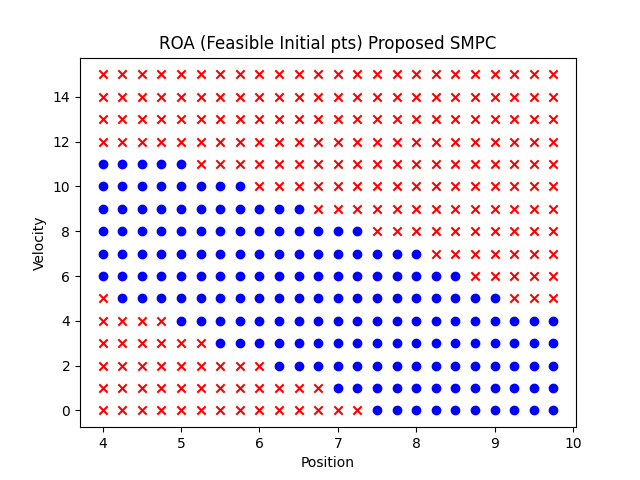}}


  \caption{Comparisons of ROA with three control policies }
  \label{fig:roa_all_4}
\end{figure}

\begin{figure}[h]

\centering
\subfigure[Open-loop policy]{\includegraphics[width=0.32\textwidth]{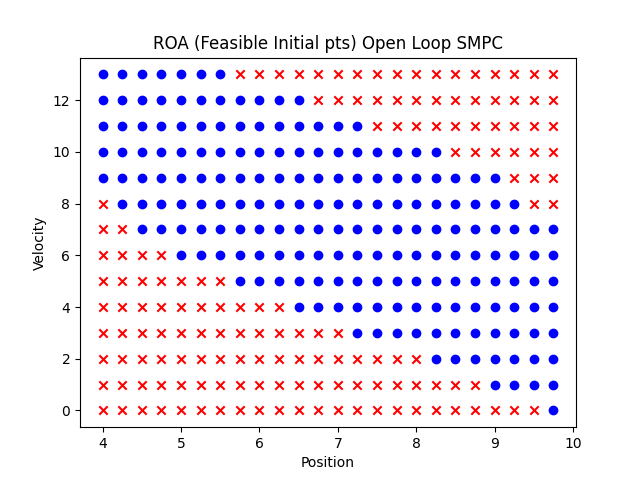}}
\subfigure[Full feedback policy]{\includegraphics[width=0.32\textwidth]{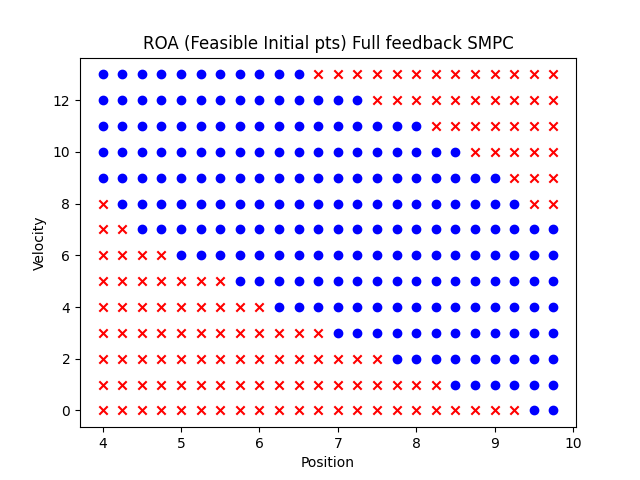}}
\subfigure[Proposed feedback policy]{\includegraphics[width=0.32\textwidth]{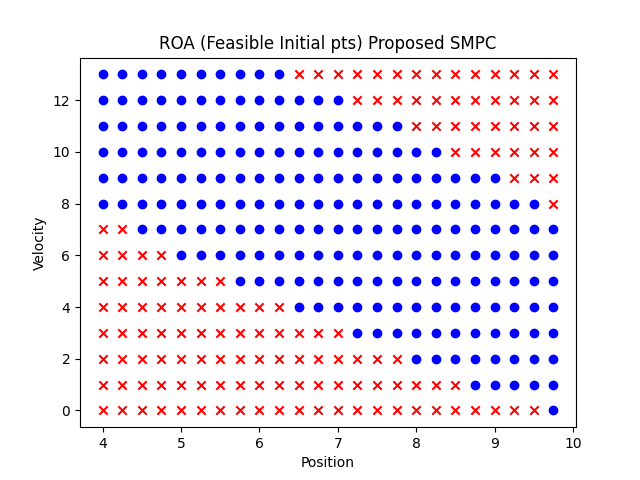}}



  \hspace{0em}
  \caption{Comparisons of ROA with three control policies for the less tight scenarios with $s_0^{\mathrm{env}}=21$}
  \label{fig:roa_all_7}
\end{figure}

Next, we run the controllers for the system from time step $0$ to time step $N_{\mathrm{task}}$ and carry out this process repeatedly for a total of $N_{\mathrm{trial}}$ trials. As a result, the accumulated closed-loop costs are illustrated in Fig. \ref{fig:clcost_over50}. 
Additionally, the gap distance between two vehicles are shown in the Fig. \ref{fig:gapdist}. In both graphs, the blue line represents the proposed policy, the cyan line represents the open-loop policy and the red line represents the full feedback policy.
The results reveal that the proposed method yields closed-loop costs that closely resemble those achieved by the full feedback policy, and notably surpass the costs incurred by the open-loop policies. Although our objective is about the nominal cost, the proposed approach has the similar average closed-loop costs with the full feedback policy and much better than the open-loop policy as you can see in the Table. \ref{table:average_costs}. Also, the gap distance graph in Fig. \ref{fig:gapdist} demonstrates a striking similarity in the outcomes derived from both the full feedback policy (red) and the proposed policy (blue). Note that the red lines and the blue lines are drawn in the wider range which is near the constraint boundary, while the cyan lines are drawn in the smaller range. The plot shows the behaviors of the proposed method are implicitly less conservative since they are much closer to the boundary of the constraints.

To check the proposed approach's robustness, we run the similar simulations for the different initial $s_0^{\mathrm{env}}$ which is $21$m. The results are shown in Fig. \ref{fig:clcost_over50_new}, Fig. \ref{fig:gapdist_new} and Table. \ref{table:average_costs_new}. Similar to the previous results, the proposed policy achieves nearly equivalent performance to the full feedback policy in terms of both ROA and the cumulative closed-loop costs.

\begin{figure}[h!]
    \centering
    \subfigure[Comparison of the closed-loop costs over 50 trials \label{fig:clcost_over50}]{\includegraphics[width=0.48\textwidth]{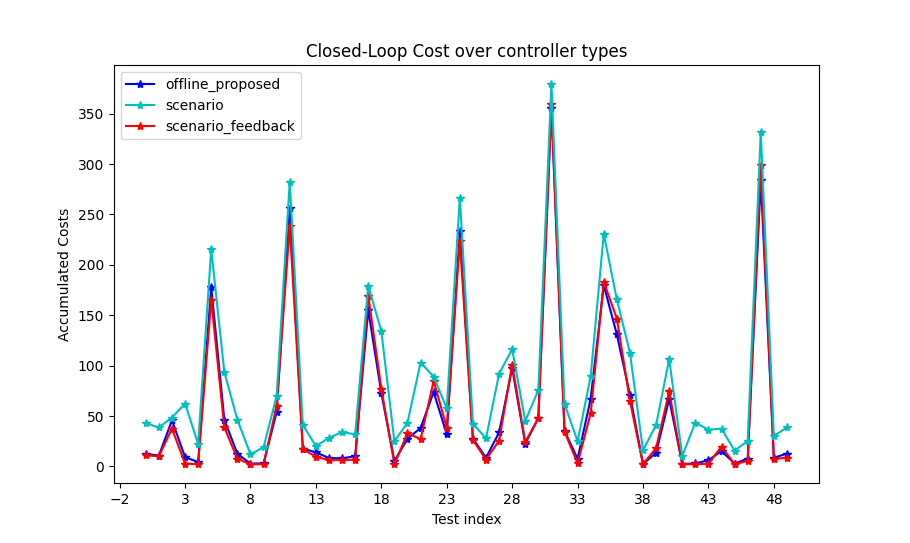}}
    \hspace{0.01\textwidth}
    \subfigure[Gap distance graph during 20 time steps with $s_0^{\mathrm{env}}=20$ \label{fig:gapdist}]{\includegraphics[width=0.48\textwidth]{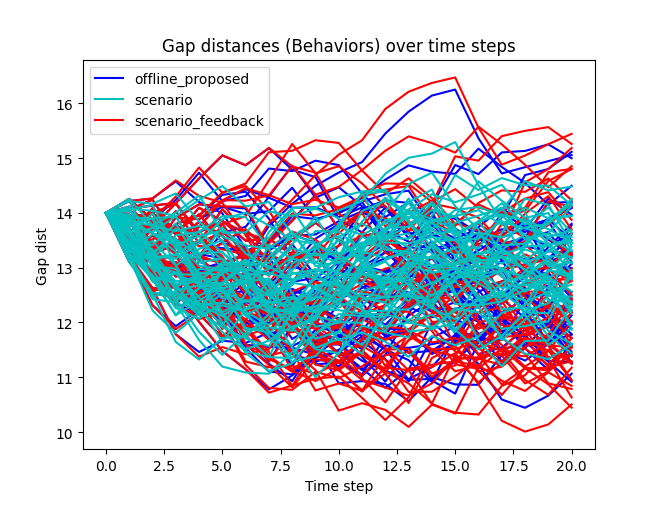}}
    \caption{(a) Comparison of the closed-loop costs over 50 trials, (b) Gap distance graph during 20 time steps with $s_0^{\mathrm{env}}=20$ \label{fig:gapdist}}
\end{figure}

\begin{figure}[h!]
    \centering
    \subfigure[Comparison of the closed-loop costs over 50 trials \label{fig:clcost_over50_new}]{\includegraphics[width=0.48\textwidth]{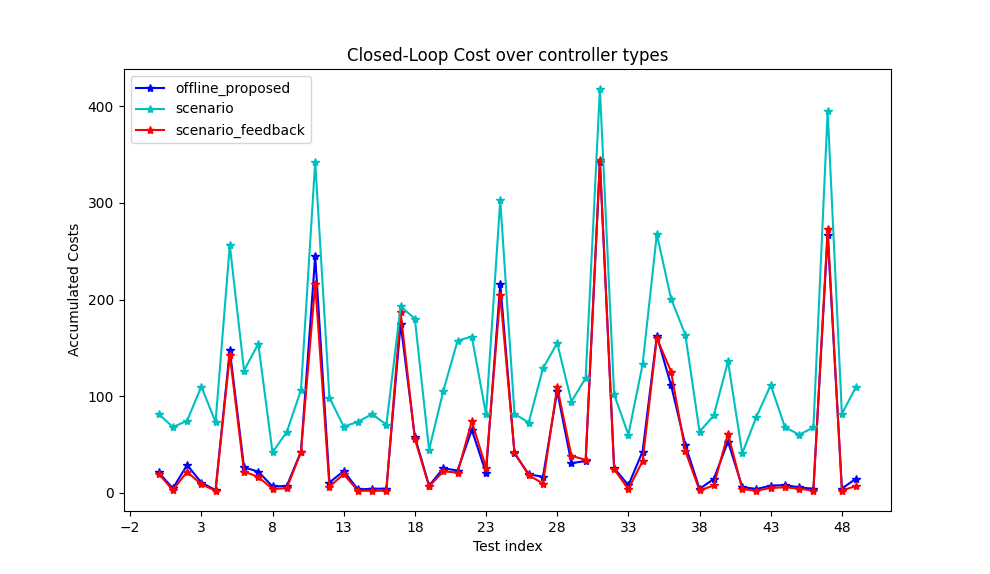}}
    \hspace{0.01\textwidth}
    \subfigure[Gap distance graph during 20 time steps with $s_0^{\mathrm{env}}=21$ \label{fig:gapdist_new}  ]{\includegraphics[width=0.48\textwidth]{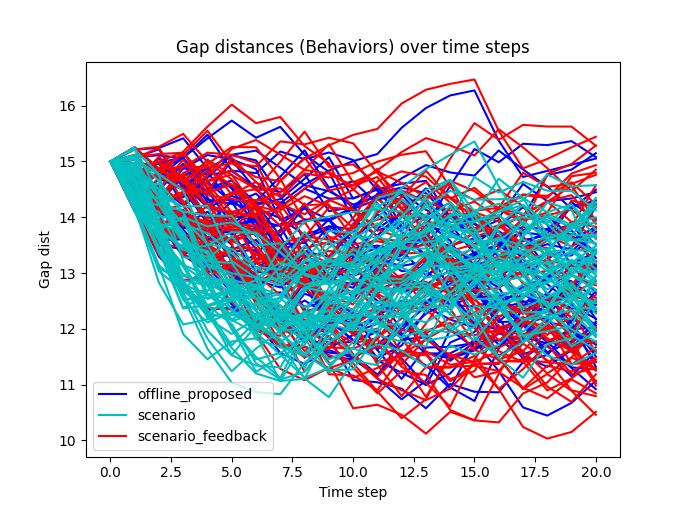}}
    \caption{(a) Comparison of the closed-loop costs over 50 trials, (b) Gap distance graph during 20 time steps with $s_0^{\mathrm{env}}=21$ \label{fig:gapdist}}
      
\end{figure}

\begin{table}[h!]
\centering
    \begin{tabular}{|c|| c| c |c| c|} 
 \hline
  & open-loop  & full feedback & proposed \\ [0.5ex] 
 \hline
 Costs &84.09& 56.02 & 56.85   \\  [0.5ex] 
 \hline

\end{tabular}
\caption{Average closed-loop costs over 50 trials with $s_0^{\mathrm{env}}=20$}
\label{table:average_costs}
\end{table}

\begin{table}[h!]
\centering
    \begin{tabular}{|c|| c| c |c| c|} 
 \hline
  & open-loop  & full feedback & proposed \\ [0.5ex] 
 \hline
 Costs &128.16 & 49.96 & 51.64   \\  [0.5ex] 
 \hline

\end{tabular}
\caption{Average closed-loop costs over 50 trials with $s_0^{\mathrm{env}}=21$}
\label{table:average_costs_new}
\end{table}

\vspace{25pt}
\subsection{Computation time}

The computation time of the proposed approach comprises: i) feature policy extraction via SVD in the offline phase, ii) computation of the approximate set in the offline phase, and iii) solving the modified MPC problem with truncated decision variables in the online phase. The optimization problems are formulated using the CasADi interface \cite{casadi} in Python, and solved with IPOPT \cite{ipopt}. All simulations and computations are carried out on a ThinkPad P53 with a 2.60 GHz Intel Core i7-9850H processor and 16GB RAM.
\begin{enumerate}[i)]
\item For $10^3, 10^4, 10^5$ samples, considered sufficient for chance constraint satisfaction \cite{calafiore2006scenario}, the offline computation times for performing SVD on the large constraint matrix and computing the approximate set are shown in Table \ref{table:offline_computation_time}. The complete offline computation for $10^5$ samples takes about 20 to 30 seconds.

\item For the online phase, using $10, 100, 1000$ samples, we compare the computation time between the full feedback policy, the open-loop policy, and the proposed method. The full feedback policy requires 0.1s to 0.2s with 100 samples, while the proposed MPC consistently completes in less than 0.01s due to pre-computed sampling processes in the offline phase (Table \ref{table:computation_time}). This highlights the significant speed advantage of the proposed method over the full feedback policy while maintaining control performance. The results justify using the proposed method in real-time applications, as it enables pre-computation of the truncated feature policy and the approximate set, achieving execution speeds at least 10X faster than the full feedback policy.
\end{enumerate}

\begin{remark}
    Although the offline computation time may seem substantial, the proposed method remain highly beneficial. Unlike high-dimensional parametrized policies such as neural network-based control policies, the proposed method preserves interpretability and tunability, while requiring less computation than extensive training processes. Also it still includes real-time optimization, enabling adaptive feedback gains in time, and is competitive with brute-force computation of all solutions. 
    Additionally, constructing the feasible set is performed offline only once. This step does not need to be repeated even if control costs change, as long as constraints are consistent. These features enhance the method’s practicality and efficiency, making it highly valuable for SMPC applications.
\end{remark}

\begin{table}[h!]
\centering
    \begin{tabular}{|c|| c| c |c| c|} 
 \hline
  \# of samples &1,000  & 10,000 &  100,000 \\ [0.1ex] 
 \hline
 SVD & 0.08& 0.3 & 2.2   \\  [0.1ex] 
 \hline
 Approx. set & 8.1 & 13.3 & 22.1 \\  [0.1ex] 
 \hline

\end{tabular}
\caption{Computation time of offline process (time unit:s)}
\label{table:offline_computation_time}
\end{table}

\begin{table}[h!]
\centering
    \begin{tabular}{|c|| c| c |c|} 
 \hline
 $\sharp$ of samples &10  & 100 & 1000\\ [0.1ex] 
 \hline
 open-loop  & 0.002 $\sim$ 0.009 & 0.031 $\sim$ 0.088 & 0.092 $\sim$ 0.171
  \\  [0.1ex] 
 \hline
 full feedback & 0.035 $\sim$ 0.049 & 0.116 $\sim$ 0.151 & 0.281 $\sim$ 1.118 \\ [0.1ex] 
  \hline
 proposed &0.003 $\sim$ 0.011 &0.003 $\sim$ 0.011  & 0.003 $\sim$ 0.011 \\ [0.1ex] 
 \hline
\end{tabular}
\caption{Computation time of online MPC for each policy (time unit: s)}
\label{table:computation_time}
\end{table}



\section{Conclusions}

We proposed a fast stochastic MPC for uncertain linear systems subject to chance constraints. Our approach involves extracting feature components from affine disturbance feedback policies using a stacked constraint matrix derived from multiple offline samples, without changing the constraint satisfaction. We then computed an approximate feasible set for the feature decision variables while maintaining a desired \(1 - \delta\) confidence level, allowing the MPC problem to be solved with reduced decision variables.
During online MPC implementation, the proposed approach solves the modified MPC problem with the offline computed sets and the truncated feature policy. With numerical simulations, we demonstrated that our approach achieves control performance comparable to full-feedback control policies in terms of Region of Attraction (ROA) and accumulated closed-loop costs. Importantly, the proposed method achieves computation speeds that are at least ten times faster, making it highly suitable for real-time applications.


\bibliography{mybib}

\end{document}